\documentclass[preprint]{elsarticle}
\usepackage{epsfig}
\usepackage{subfigure}
\usepackage{latexsym}
\usepackage{graphicx}
\usepackage{amssymb}
\usepackage{amsthm}

\newtheorem{theorem}{Theorem}[section]
\newtheorem{corollary}{Corollary}[section]
\newtheorem{lemma}{Lemma}[section]
\newtheorem{proposition}{Proposition}[section]
\newtheorem{definition}{Definition}[section]

\makeatletter
\def\bdiv{%
\nonscript\mskip-\medmuskip\mkern5mu%
\mathbin{\operator@font div}\penalty900\mkern5mu%
\nonscript\mskip-\medmuskip}
\makeatother

\input epsf.tex

\def\angle#1{\left\langle #1\right\rangle}
\newcommand{\rank}{\mbox{\sf rank}}
\newcommand{\select}{\mbox{\sf select}}
\newcommand{\access}{\mbox{\sf access}}
\newcommand{\excess}{\mbox{\sf excess}}
\newcommand{\parent}{\mbox{\sf parent}}
\newcommand{\depth}{\mbox{\sf depth}}
\newcommand{\open}{\mbox{\sf open}}
\newcommand{\close}{\mbox{\sf close}}
\newcommand{\nextexcess}{\mbox{\sf next-excess}}
\newcommand{\levelancestor}{\mbox{\sf level-ancestor}}
\newcommand{\levelsuccessor}{\mbox{\sf level-successor}}
\newcommand{\isancestor}{\mbox{\sf is-ancestor}}
\newcommand{\firstchild}{\mbox{\sf first-child}}

\def\ceil#1{\left\lceil #1\right\rceil}
\def\floor#1{\left\lfloor #1\right\rfloor}
\def\calp{\mbox{$\cal P$}}

\def\th{\mathit th}
\newcommand{\Rank}{\mbox{\sf p-rank}} 
\newcommand{\Select}{\mbox{\sf select}} 
\newcommand{\fullrank}{\mbox{\sf rank}}

\newcommand{\Under}{\mbox{\sf under}}

\begin{document}

\begin{frontmatter}

\title{Succinct Representations of Permutations and Functions\tnoteref{t1,t2}}
\tnotetext[t1]{Work supported in part by UISTRF project 2001.04/IT.}
\tnotetext[t2]{Preliminary versions of these results have appeared in the {\it Proceedings of International Colloquium on Automata, Languages and Programming (ICALP)} in 2003 and 2004.}

\author[UW]{J. Ian Munro}
\ead{imunro@uwaterloo.ca}

\author[UL]{Rajeev Raman}
\ead{rr29@leicester.ac.uk}

\author[IMSc]{Venkatesh Raman}
\ead{vraman@imsc.res.in}

\author[SNU]{S. Srinivasa Rao\corref{cor1}}
\ead{ssrao@cse.snu.ac.kr}

\cortext[cor1]{Corresponding author}

\address[UW]{School of Computer Science, University of Waterloo, Waterloo, ON, N2L 3G1, Canada} 
\address[UL]{Department of Computer Science, University of Leicester, Leicester, LE1 7RH, UK} 
\address[IMSc]{Institute of Mathematical Sciences, Chennai, 600 113, India} 
\address[SNU]{School of Computer Science and Engineering, Seoul National University, Seoul, 151-744, Republic of Korea}


\begin{abstract}
We investigate the problem of succinctly representing an arbitrary
permutation, $\pi$, on $\{0,\ldots,n-1\}$ so that $\pi^k(i)$ 
can be computed quickly for any $i$ and any (positive or negative)
integer power $k$. A representation taking $(1+\epsilon) n \lg n + O(1)$
bits suffices to compute arbitrary powers in constant time, for any positive
constant $\epsilon \le 1$.  
A representation taking the optimal $\ceil{\lg n!} + o(n)$ bits can be
used to compute arbitrary powers in $O(\lg n / \lg\lg n)$ time.

We then consider the more general problem of succinctly representing 
an arbitrary function, $f: [n] \rightarrow [n]$ so that $f^k(i)$ can 
be computed quickly for any $i$ and any integer power $k$. We give a 
representation that takes $(1+\epsilon) n \lg n + O(1)$ bits, for any 
positive constant $\epsilon \le 1$, and computes arbitrary positive 
powers in constant time.  It can also be used to compute $f^k(i)$, 
for any negative integer $k$, in optimal $O(1+|f^k(i)|)$ time.

We place emphasis on the \emph{redundancy}, or the space beyond 
the information-theoretic lower bound that the data structure uses 
in order to support operations efficiently.   A number of lower bounds
have recently been shown on the redundancy of data structures.  
These lower bounds confirm the space-time optimality of some of 
our solutions.  Furthermore, the redundancy of one of our 
structures ``surpasses'' a recent lower bound by 
Golynski~[Golynski, SODA 2009], thus demonstrating the 
limitations of this lower bound.
\end{abstract}


\end{frontmatter}

\section{Introduction}

For an arbitrary function $f$ from $[n] = \{0,\ldots,n-1\}$ 
to $[n]$, define $f^{k}(i)$, for all $i \in [n]$, and any integer $k$
as follows: 
$$f^k(i) = \left\{ 
		  \begin{array}{ll}
			i & \mbox{ when } k=0 \\
			 f(f^{k-1}(i)) & \mbox{ when } k>0 \mbox{ and } \\
			 \{j | f^{-k}(j) = i \} & \mbox{ when } k<0. 
			 \end{array}
		  \right. $$
We consider the following problem: we are given a specific and
arbitrary (static) function $f$ from $[n]$ to $[n]$
that arises in some application.
We want to represent $f$ (after pre-processing $f$)
in a data structure that, given $k$ and $i$ as parameters, 
rapidly returns the value of $f^{k}(i)$.  For the sake of
simplicity, in the rest of the paper we assume that the given 
number $k$ is bounded by some polynomial in $n$.

Our interest is in \emph{succinct}, or highly-space efficient, 
representations of such functions, whose space usage is close 
to the information-theoretic lower bound for representing such
a function.  Since there are $n^n$ functions from 
$[n]$ to $[n]$, such a function cannot be represented
in less than $\lceil n \lg n \rceil$ bits\footnote{$\lg$ denotes the %
logarithm base 2.}.  
Any amount of memory used by a data structure that represents such a function,
above and beyond this lower bound, is termed the \emph{redundancy} 
of the data structure. We also consider the case where 
$f$ is given as a ``black box'',
i.e. the data structure is given access to a routine to evaluate $f(i)$
for any $i \in [n]$; in this case any amount of memory whatsoever used
by the data structure is its {redundancy}.
The fundamental aim is to understand precisely the minimum redundancy 
required to support operations rapidly.  

Clearly, the above problem is trivial if space is not an issue. 
To facilitate the computation in constant time, one could 
store $f^k(i)$ for all $i$ and $k$ ($|k| \le n$, along with some 
extra information), but that would require $\Omega(n^2)$ words 
of memory. The most natural compromise is to retain the values 
of $f^k(i)$ where $2 \le k \le n$ is a power of $2$. This 
$\Theta(n \lg n)$-word representation easily yields a
logarithmic evaluation scheme.  Unfortunately, this representation 
not only uses non-linear space (and is relatively slow) but also 
does not support queries for the negative powers of $f$ efficiently. 
Given $f$ in a natural representation --- the sequence $f(i)$ 
for $i = 0,\dots, n-1$, or as a black box --- 
a highly space-efficient solution is to store no additional 
data structures (zero redundancy), and
to compute $f^k(i)$ in $k$ steps, for positive $k$. 
However, this is unacceptably slow for large $k$, and still 
does not address the issue of negative powers.

\subsection{Results}

Our results are primarily in the unit-cost RAM with word size
$\Theta(\log n)$ bits, where we measure the running time and
the bits of space used by an algorithm.  We also consider
the ``black-box'' model, known also as the \emph{systematic} model \cite{GP07},
where we look at the number of evaluations of $f$ in addition
to the running time and space (in bits) used by the algorithm.  
Lower bound results are discussed in either the black-box model or in the
\emph{cell-probe} model, where we consider the
space (in bits) used by the algorithm, and the running time is
the number of $w$-bit words of the data structure 
read by the algorithm to answer a query (and all other
computation is for free).  Finally, we also briefly consider
the \emph{bit-probe} model, which is the cell-probe model with $w=1$ \cite{bitprobe}.

\subsubsection{Permutations}

We begin by considering a special case, where the function is
a permutation (abbreviated hereafter as a {\it perm\/} 
\cite{Knuth}) of $ [n] = \{0,\ldots,n-1\}$.  This turns out not
only to be an interesting sub-case in its own right, but is also 
essential to our solution to the general problem.  Note that
for storing perms, the information-theoretic lower bound
is $\calp(n) = \lceil \lg n! \rceil \approx n \lg n - 1.44 n$ bits, so
the obvious representation (as an array storing $\pi(i)$ for $i=1,\ldots,n$) 
has redundancy $\Theta(n)$ bits (and of course does not support 
inverses or powers). We obtain the following results for representing perms:
\begin{enumerate}

\item We give a representation that uses $\calp(n) + O(n (\lg\lg n)^5 / (\lg n)^2)$ 
bits, and supports $\pi()$ and $\pi^{-1}()$ in $O(\lg n/\lg\lg n)$ time.

\item In the ``black box'' model, where access to the
perm is only through the $\pi()$ operation, we show how to
support $\pi^{-1}()$ in $O(t)$ time and at most $t+1$ evaluations of $\pi()$,
using $(n/t) (\lg n + \lg t + O(1))$ bits, for any $1 \le t \le n$.

\item Given a structure that represents a perm $\pi$ in space $S(n)$ bits,
and supports $\pi()$ and $\pi^{-1}()$ in time $t_f(n)$ and $t_i(n)$ 
respectively, we show how to represent a given perm $\pi'$ on $[n]$ in
space $S(n) + O(n \lg n/ \lg \lg n)$ bits (or
$S(n) + O(\sqrt{n} \lg n)$ bits) and support 
arbitrary powers of $\pi'$ in $t_f(n) + t_i(n) + O(1)$ time
(or $t_f(n) + t_i(n) + O(\lg \lg n)$ time, respectively).
\end{enumerate}
As corollaries, we get the following representations of perms:
%
\begin{enumerate}
\setcounter{enumi}{3}
\item one that uses ${\cal P}(n) + O((n/t) \lg n)$ bits, and
supports $\pi()$ in $O(1)$ time and $\pi^{-1}()$ in $O(t)$ time,
for any $t \le \lg n$.

\item one that uses ${\cal P}(n) + O((n/t) \lg n)$ bits and supports
$\pi^k()$ in $O(t)$ time for arbitrary $k$, for any 
$t \le \lg n$.

\item one that uses ${\cal P}(n) + O(n (\lg \lg n)^5/(\lg n)^2)$ bits and 
supports $\pi^k()$ in $O(\lg n / \lg \lg n )$ time for arbitrary $k$.
\end{enumerate}


\subsubsection*{Related Work}

Perms are fundamental in computer 
science and have been the focus of
extensive study.  A number of papers have dealt with issues pertaining
to perm generation, membership in perm groups etc.  
There has also been work on space-efficient representation 
of restricted classes of perms, such as the perms representing 
the lexicographic order of the suffixes of a 
string \cite{GV00,HMR}, or so-called approximately min-wise
independent perms, used for document similarity estimation \cite{BFCM}.
Our paper is the first to study the space-efficient
representation of general perms so that general powers can be computed
efficiently (however, see the discussion on Hellman's work 
in Section~\ref{subsec:motivation}).

Recently Golynski \cite{golynski-09,golynski-thesis} showed a number
of lower bounds for the redundancy of permutation representations.
He showed a space lower bound of $\Omega((n/t)\lg (n/t))$ bits for 
Item (2) for any algorithm that evaluates $\pi$ at most $t < n/2$ 
times \cite[Theorem 17]{golynski-thesis}.
Thus, (2) is asymptotically optimal for all $t = n^{1-\Omega(1)}$. 
Furthermore, Golynski \cite{golynski-09} showed that the redundancy
of (4) is asympotically optimal in the cell probe model with 
word size $w = \lg n$: specifically,
that any perm representation which supports $\pi()$ in $O(1)$ probes
and $\pi^{-1}()$ in $t$ probes, for any
$t \le (1/16)(\lg n /\lg \lg n)$, must
have asymptotically the same redundancy as (4). He also shows
that any perm that supports both $\pi()$ and $\pi^{-1}()$ in
at most $t$ cell probes, for any $t \le (1/16)(\lg n /\lg \lg n)$, must
have redundancy $\Omega(n (\lg \lg n)^2 / \lg n)$. 
In the preliminary version of this paper \cite{MRRR}, a
perm representation was given that supported $\pi()$ and
$\pi^{-1}()$ in $O(\lg n/\lg \lg n)$ time, and had
redundancy $\Theta(n (\lg \lg n)^2 / \lg n)$.  Golynski suggested that
the result of \cite{MRRR} was ``optimal up to constant factor in the cell
probe model''. However, we note that the lower bound is quite sensitive to the
precise constant in the number of probes: our result (1) obtains an asymptotically smaller
redundancy by using over $2 \lg n/\lg \lg n$ cell probes.

\subsubsection{Functions}

For general functions from $[n]$ to $[n]$, 
our main result is that we reduce the problem of 
representing functions to that of representing
permutations, with $O(n)$ additional bits.  As corollaries,
we get the following representations of functions, both of 
which use close to the information-theoretic minimum amount of
space, and answer queries in optimal time:
%
\begin{enumerate}
\item one that uses $n \lg n (1 + 1/t) + O(1)$ bits, and
supports $f^{k}(i)$ in $O(1 + |f^{k}(i)| \cdot t)$ time for any integer $k$, 
and for any $t \le \lg n / \lg \lg n$.

\item one that uses $n \lg n + O(n)$ bits and 
supports $f^{k}(i)$ in $O((1 + |f^{k}(i)|) \cdot (\lg n / \lg \lg n))$ time, 
for any integer $k$.
\end{enumerate}
Along the way, we show that an unlabelled static $n$-node rooted tree can be
represented using the optimal $2n+o(n)$ bits of space to answer
\emph{level-ancestor} --- given a node $x$ and a number $k$, to report the $i$-th ancestor 
of $x$ --- and \emph{level-successor/level-predecessor} queries ---  to report the 
next/previous node at the same level as the given node --- 
in constant time.   We represent the tree in $2n$ bits as a balanced parenthesis (BP) sequence.
The key technical contribution is to provide a $o(n)$-bit
index for \emph{excess search} in a BP sequence.  For a position $i$
in a BP sequence, $\excess(i)$ is the number of unclosed open parentheses up to that
position (this corresponds to the depth of a node in the tree represented by the BP).
The operation $\nextexcess(i,k)$, starting at a position $i$ in the BP sequence,
finds the next position $j$ whose excess is $k$; 
we support $\nextexcess$ in $O(1)$ time provided that
$j$'s excess is at most $(\lg n)^{c}$ below or above the excess of $i$ (i.e., $|k - \excess(i)| = O((\lg n)^{c})$), for any fixed constant $c \ge 0$.
To add standard navigational operations, one can use existing $o(n)$ bit indices
for BP sequences \cite{MR}.

\subsection*{Related work}
 
The problem of representing a function $f$ space-efficiently
in the ``black box'' model, so that $f^{-1}$ can be computed quickly,
was considered by Hellman \cite{hellman}.
Specialized to perms, Hellman's idea is similar to our
``black box'' representation for representing a perm and its inverse,
modulo some implementation details.
The version of the function powers problem that we consider is different:
whereas Hellman attempts, given $x$, to find any $y$ such that $f(y) = x$,
we enumerate all such $y$.  Furthermore, our solution does not use the
``black box'' model, and assumes space for representing $f$ in its entirety,
which is both unnecessary and prohibitive in Hellman's context.

Representing trees to support level-ancestor queries is a well-studied problem.
Solutions with $O(n)$ preprocessing time and $O(1)$ query time were given by Dietz
\cite{Dietz}, Berkman and Vishkin \cite{BV} and by Alstrup and Holm
\cite{AH}. A much simpler solution was given by Bender and
Farach-Colton \cite{BF}.  For a tree on $n$ nodes, all these solutions
require $\Theta(n)$ words, or $\Theta(n \lg n)$ bits, to represent the tree itself, and the additional
data structures stored to support level-ancestor queries also take $\Theta(n)$ words 
(level-successor/predecessor is trivial using $\Theta(n)$ words).

As noted above, our interest is in succinct tree representations.  We make a few
remarks about such representations, so as to better understand our contribution in
relation to others. 
Succinct tree representations can also be considered to be split into
a \emph{tree encoding} that takes $2n + o(n)$ bits, and an \emph{index}
of $o(n)$ bits for that tree encoding.  There are many tree encodings,
including BP \cite{MR}, DFUDS \cite{bdmrrr}, LOUDS \cite{Jacobson} and Partition \cite{GRR}, and it is not known
if they are equivalent, i.e. if there are operations that have $o(n)$ sized 
indices for one tree encoding and not the other.  Another feature is that
different tree encodings impose different numberings on the nodes of the tree.
Therefore, a result showing a succinct index for a particular operation in (say) BP
does not imply the existence of a succinct index for that operation in (say) LOUDS.
This matters from an application perspective because the only way to get a space-efficient
data structure that simultaneously supports operations $a$ and $b$, where
$a$ and $b$ are known to be supported only
by (say) LOUDS and BP-based tree encodings respectively, would be to encode the tree twice,
once each in LOUDS and BP  
and to maintain the correspondence between the LOUDS and BP numberings,
which would severely affect the space usage.

We provide $o(n)$-bit BP indices for the operations of level-ancestor and 
level-successor/predecessor, via excess search.  Geary et al. \cite{GRR} gave 
a $o(n)$-bit index for supporting level-ancestor in $O(1)$ time using the 
Partition encoding, but they did not provide support for level-successor/predecessor;
a $o(n)$-bit index for supporting these queries was announced by He et al. \cite{HMR-trees}.  
Very recently Sadakane and Navarro \cite{SN10} gave an alternative algorithm for excess 
search in BP and showed that excess search together with range-minimum queries suffice 
to support a wide variety of tree operations, among other things.  Their excess 
index is of smaller size,  but seems
not to support search for excess values greater than the starting point. 

\subsection{Motivation}
\label{subsec:motivation}

There are a number of motivations for succinct data structures in
general, many to do with text indexing or representing huge graphs
\cite{GV00,Jacobson,MR,RRR-SODA}.  
Work on succinct representation of a perm and its inverse was, for one
of the authors, originally motivated by a data warehousing
application. Under the indexing scheme in the system, the perm
corresponding to the rows of a relation sorted under any given key was
explicitly stored. It was realized that to perform certain joins, the
inverse of a segment of this perm was precisely what was required. The
perms in question occupied a substantial portion of the several
hundred gigabytes in the indexing structure and doubling this space
requirement (for the perm inverses) for the sole purpose of improving
the time to compute certain joins was inappropriate.  

Since the publication of the preliminary versions of these papers, the
results herein have found numerous applications, most notably to
the problem of supporting $\fullrank$ and $\Select$ operations over
strings of large alphabets \cite{Golynskietal}. Other
applications arise in Bioinformatics \cite{BYM}.  The more general
problem of quickly computing $\pi^{k}()$ also has number of
applications.  An interesting one is determining the $r^{\th}$ root of
a perm \cite{Pouyanne}.  Our techniques not only solve the $r^{\th}$
power problem immediately, but can also be used to find the $r^{\th}$
root, if one exists. Inverting a ``one-way'' function, particularly in the scenario
considered by Hellman \cite{hellman}, is a fundamental task in cryptography.

Finally, very recently a number of results have been shown that
focus on the redundancy of succinct data structures for various objects,
including \cite{GP07,golynski-07,golynski-09,PV10}; we have already
mentioned lower bounds on the redundancy of 
representing perms in particular.  This has been accompanied by
some remarkable
results on very low-redundancy data structures.  For example,
consider the simple task of representing a sequence of
$n$ integers from $[r]$, for some $r \ge 1$ to permit
random access to the $i$-th integer.  The naive bound of 
$n \ceil{\lg r}$ bits has redundancy $\Theta(n)$ bits
relative to the optimal $\ceil{n \lg r}$ bits.  Following the
first non-trivial result on this topic (\cite[Theorem~3]{MRRR}),
a line of work culminated in Dodis et al.'s remarkable result that $O(1)$-time
access can be obtained with effectively zero redundancy \cite{DPT-10}. 
We also note that the redundancy is often important in practice,
as the ``lower-order'' redunancy term in the space usage is
often significant for practical input sizes \cite{Geary-BP}.

The remainder of the paper is organized as follows. The next section
describes some previous results on indexable dictionaries used in
later sections.  
 Section~\ref{sec:perms} deals  with permutation representations.
In Section~\ref{sec:shortcut} we describe the `shortcut' method, and 
Section \ref{sec:benes} describes an optimal space
representation based on Benes networks. Both of these are representations
supporting $\pi()$ and $\pi^{-1}()$ queries, and we consider the optimality
of these solutions in Section~\ref{subsec:optimality}.
In Section~\ref{sec:powers1} we consider representations that support arbitrary powers.
Sections \ref{sec:level-anc} and \ref{sec:functs} deal with general function
representation. 
Section \ref{sec:level-anc} outlines new operations on balanced parenthesis sequences
which lead to 
an optimal-space tree representation that supports level-ancestor queries
along with various other navigational operations in constant time. 
Section \ref{sec:functs} describes a succinct representation of a function that
supports computing arbitrary powers in optimal time.

\section{Preliminaries}\label{sec:prelim}
Given a set $S \subseteq [m]$, $|S|=n$,  define the following operations:
\begin{description}
\item[$\fullrank(x,S)$:]  Given $x \in [m]$, return $|\{y \in S | y < x \}|$,
\item[$\Select(i, S)$:] Given $i \in [n]$, return the $i+1$-st smallest element in $S$,
\item[$\Rank(x, S)$:] Given $x \in [m]$,
 return $-1$ if $x \not \in S$ and $\fullrank(x,S)$ otherwise (the \emph{partial rank} operation).  
\end{description}
Furthermore, define the following data structures:

\begin{itemize}
\item 
A {\em fully indexable dictionary (FID)} representation for $S$ supports
$\fullrank(x,S)$, $\Select(i,S)$, $\fullrank(x,\bar{S})$ and $\Select(i,\bar{S})$
in $O(1)$ time.

\item An {\em indexable
 dictionary (ID)} $S$ supports $\Rank(x,S)$ and $\Select(i,S)$ in $O(1)$ time.
\end{itemize}
Raman, Raman and Rao \cite{RRR-SODA} show the following:
\begin{theorem} On the RAM model with wordsize $O(\lg m)$ bits:
\label{litomfid}
\begin{itemize}
\item[(a)]
There is a FID for a set $S \subseteq [m]$ of size $n$
using at most $\ceil {\lg {m \choose n}} + O(m \lg \lg m/\lg m)$ bits.
\item[(b)]
There is an ID for a set $S \subseteq [m]$ of size $n$
using at most $\ceil {\lg {m \choose n}} + o(n) + O(\lg \lg m)$ bits.
\end{itemize}
\end{theorem}

\section{Representing Permutations}
\label{sec:perms}

\subsection{The Shortcut Method}\label{sec:shortcut}

We first provide a space-efficient representation (based on Hellman's idea) that supports $\pi^{-1}()$ in the ``black box'' model.  Recall that
in the ``black box'' model, the perm is accessible only
through calls of $\pi()$.
Let $t \geq 2$ be a parameter. We trace the cycle structure
of the perm $\pi$, and for every cycle whose length $k$ is greater
than $t$, the key idea is to associate
with some selected elements, a \emph{shortcut pointer} to an element
$t$ positions prior to it.  Specifically, let $c_0, c_1, \ldots, c_{k-1}$
be the elements of a cycle of the perm $\pi$ such that
$\pi(c_i)=c_{{(i+1)}\bmod{k}}$, for $i=0,1, \ldots, k-1$.
We associate shortcut pointers with the indices
whose $\pi$ values are $c_{i t}$, for $i = 0,1, \ldots, l = \floor {k/t}$,
and the shortcut pointer value at
$c_{it}$ stores the index whose $\pi$ value is $c_{((i-1) \bmod (l+1))t}$,
for $i = 0,1, \ldots, l$ (see Fig.~\ref{fig:shortcut}).  Let $s \leq
n/t$ be the number of shortcut pointers after doing this for every
cycle of the perm and let $d_1 < d_2 < \ldots < d_s$ be the elements
associated with shortcut pointers.

\begin{figure}[htbp]
\centering
\epsfxsize 12cm
\epsfbox{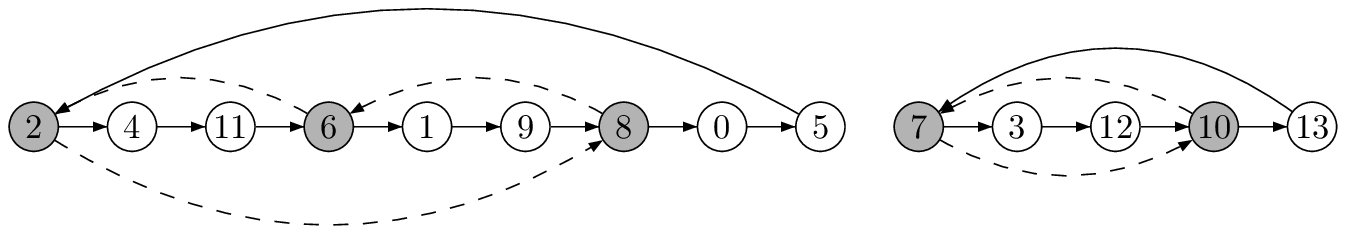}
\caption{Shortcut method. Solid lines denote the perm, %
and the dotted lines denote the shortcut pointers. The shaded %
nodes indicate the positions having shortcut pointers. }
\label{fig:shortcut}
\end{figure}

We store the set $\{d_i\}$ in a data structure $D$ that is
an instance of the indexable dictionary (ID) of
Theorem~\ref{litomfid}(b). Given an index $i$, $D$ allows us
to test if a particular element has a shortcut pointer with it, and if so, returns its position in the
set $\{d_i\}$.  We store the sequence $\{s_i\}$, where $s_i$ is the
shortcut pointer associated with $d_i$ in an array $S$.
The following procedure computes
$\pi^{-1} (x)$ for a given $x$:
\begin{tabbing}
xxxx\=xxxx\=xxxx\=xxxx\=xxxx\=xxxx\=xxxx\=xxxx\=xxxx\=xxxx\=xxxx\=xxxx\=xxxx\=\kill

$i := x$;\\
\>{\bf while} $\pi (i) \neq x$ {\bf do}\\
\>\> {\bf if} $i \in D$ and $\Rank(i,D) = r$  \>\>\>\>\>\>\>$/\!/\,$both found by querying $D$ \\
\>\>\>{\bf then} $j := S[r]$;\\
\>\>\>{\bf else} $j := \pi (i)$;\\
\>\>$i := j;$\\
\>{\bf endwhile}\\
{\bf return $i$}
\end{tabbing}
Since we have a shortcut pointer for every $t$ elements of a cycle,
the number of $\pi()$ evaluations made by the algorithm is at most
$t+1$, and all other operations take $O(1)$ time by Theorem~\ref{litomfid}.
By the standard approximation
$\lceil \lg {n \choose s} \rceil = s (\lg (n/s) + O(1))$, we see that
the space used by $D$ is at most $(n/t) (\lg t + O(1))$ bits.  The space
used by $S$ is clearly $s \lceil \lg n \rceil = s (\lg n + O(1))$.
Thus we have:

\begin{theorem}
\label{thm:shortcuts2}
Given an arbitrary permutation $\pi$ on $[n]$ as a ``black box'',
and an integer $1 \le t \le n$, there is a data structure that
uses at most $(n/t) (\lg n + \lg t + O(1))$ bits that allows
$\pi^{-1} ()$ to be computed in at most $t+1$ evaluations of $\pi()$,
plus $O(t)$ time.
\end{theorem}
We get the following easy corollary:
\begin{corollary}
\label{cor1}
There is a representation of an arbitrary perm $\pi$ on $[n]$ using at most
${\calp}(n) + O((n/t) \lg n))$ for any $1 \le t \le \lg n$ that supports
$\pi()$ in $O(1)$ time and $\pi^{-1}$ in $O(t)$ time.
\end{corollary}
\begin{proof}
We represent $\pi$ naively as an array taking
$n \lceil \lg n \rceil$  = ${\calp}(n) + O(n)$ bits,
and allowing $\pi()$ to be computed in $O(1)$ time,  and apply
Theorem~\ref{thm:shortcuts2}.  The space bound follows since for $t \le \lg n$,
$(n/t)(\lg n + \lg t + O(1)) = \Omega(n)$.
\end{proof}
\smallskip
\noindent
\emph{Remark:} Choosing $t = \lceil (1/\epsilon) \rceil$ for any
constant $\epsilon > 0$ in Corollary~\ref{cor1} we get a representation
of a permutation $\pi$ on $[n]$ in $(1+\epsilon)  n \lg n$ bits where
$\pi()$ and $\pi^{-1}$ both take $O(1)$ time.
\smallskip

\subsection{Representations based on the Benes network}
\label{sec:benes}

\subsubsection{The Benes Network}

The results in this section are based on the Benes network, a
communication network composed of a number of \emph{switches},
which we now briefly outline (see \cite{LeightonBook} for details).
Each switch has two inputs $x_0$ and $x_1$ and two
outputs $y_0$ and $y_1$ and can be
configured either so that $x_0$ is connected to $y_0$ (i.e. a packet
that is input along $x_0$ comes out of $y_0$) and $x_1$ is connected
to $y_1$, or the other way around.  An $r$-Benes network has $2^r$
inputs and $2^{r}$ outputs, and is defined as follows.  For $r=1$, the Benes
network is a single switch with two inputs and two outputs.  An
$(r+1)$-Benes network is composed of $2^{r+1}$ switches and two
$r$-Benes networks, connected as shown in Fig.~\ref{fig:benes}(a).
A particular setting of the switches of a Benes network {\it
realises\/} a perm $\pi$ if a packet introduced at input $i$ comes out
at output $\pi(i)$, for all $i$ (Fig.~\ref{fig:benes}(b)).  The
following properties are either easy to verify or well-known
\cite{LeightonBook}.

\begin{itemize}
\item An $r$-Benes network has $r 2^r - 2^{r-1}$ switches, and every
path from an input to an output passes through $2r-1$ switches;
\item For every perm $\pi$ on $[2^r]$ there is a setting of the
switches of an $r$-Benes network that realises $\pi$.
\end{itemize}
\begin{figure}[htbp]
\begin{center}
\epsfxsize \textwidth \epsfbox{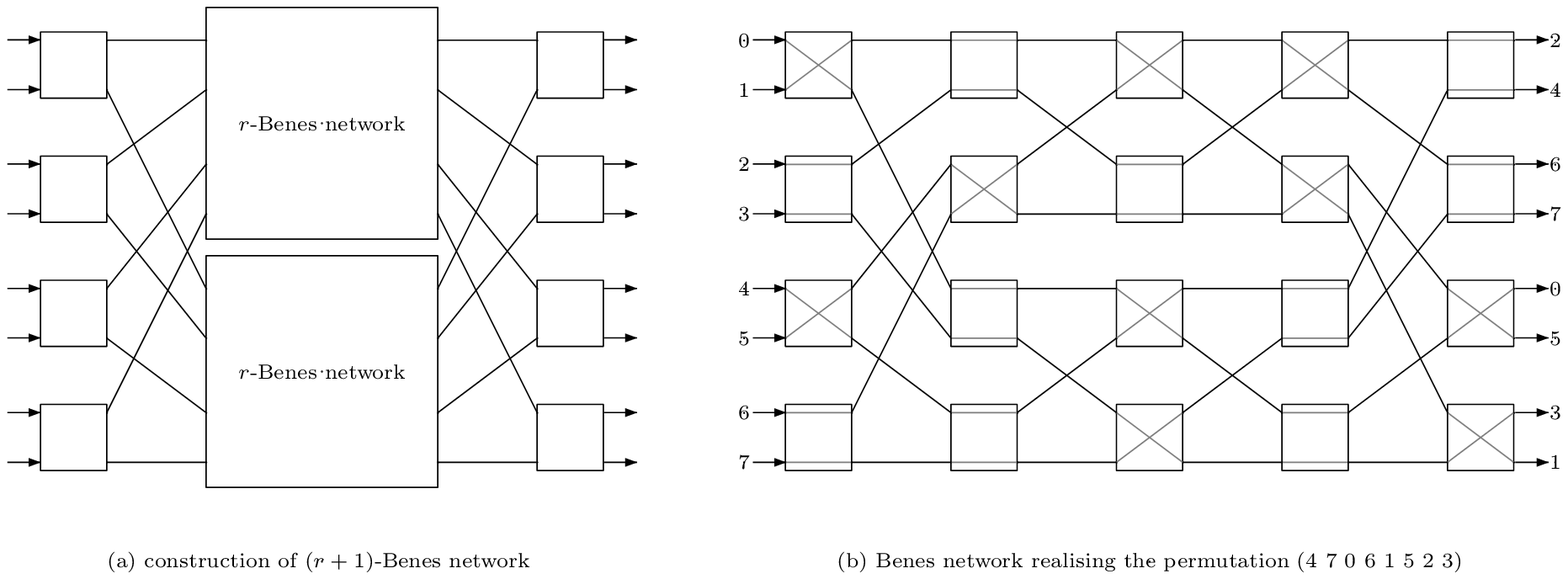}
\end{center}
\caption{The Benes network construction and an example}
\label{fig:benes}
\end{figure}

Clearly, Benes networks may be used to represent perms.
If $n = 2^r$, a representation of a perm
$\pi$ on $[n]$ may be obtained by configuring an $r$-Benes network to
realize $\pi$ and then listing the settings of the switches in some
canonical order (e.g. level-order).  This represents $\pi$ using $r
2^r - 2^{r-1}$ = $n \lg n - n/2$ bits.
Given $i$, one can trace the path taken by a packet
at input $i$ by inspecting the appropriate bits in
this representation, and thereby compute $\pi(i)$; by
tracing the path back from output $i$ we can likewise
compute $\pi^{-1}(i)$.  The time taken is clearly
$O(\lg n)$; indeed, the algorithm only makes
$O(\lg n)$ bit-probes.  To summarize:
\begin{proposition}
\label{prop:simple}
When $n = 2^r$ for some integer $r>0$, there is a representation of
an arbitrary perm $\pi$ on $[n]$ that uses $n \lg n - n/2$ bits and
supports the operations $\pi()$ and $\pi^{-1}()$ in $O(\lg n)$ time.
\end{proposition}
However, the Benes network
has two shortcomings from our viewpoint: firstly, the Benes network
is defined only for values of $n$ that are powers of 2.  In order
to represent a perm with $n$ not a power of 2, rounding up $n$
to the next higher power of 2 could double the space usage, which
is unacceptable.  Furthermore, even for $n$ a power of 2, representing
a perm using a Benes network uses ${\cal P}(n) + \Omega(n)$ bits.

We now define a family of Benes-like networks that admit greater
flexibility in the number of inputs, namely the $(q,r)$-Benes
networks, for integers $r\ge 0, q >1$.
\begin{definition}
A {\it $q$-permuter\/} to be a communication network that has $q$ inputs
and $q$ outputs, and realises any of the $q!$ perms of its
inputs (an $r$-Benes network is a $2^r$-permuter).
\end{definition}
\begin{definition}
A $(q, r)$-Benes network is a $q$-permuter for $r = 0$, and for $r > 0$ it
is composed of $q 2^r$ switches and two $(q, r-1)$-Benes networks,
connected together in exactly the same way as a standard Benes network.
\end{definition}

\begin{lemma}
\label{lem:modben}
Let $q > 1, r \ge 0$ be integers and take $p = q 2^r$.  Then:
\begin{enumerate}
\item A $(q,r)$-Benes network consists of $q 2^{r-1} (2r-1)$ switches and
$2^{r}$ $q$-permuters;
\item For every perm $\pi$ on $[p]$ there is a setting of the switches
of a $(q,r)$-Benes network that realises $\pi$.
\end{enumerate}
\end{lemma}

\begin{proof}
(1) is obvious; (2) can be proved in the same way as for a standard
Benes network. 
\end{proof}

We now consider representations based on $(q,r)$-Benes networks; a
crucial component is the representation of the central $q$-permuters, which
we address in the next subsection.
Since we are not interested in designing communication networks as such,
we focus instead on ways to represent the perms represented
by the central $q$-permuters in
optimal (or very close to optimal) space and operate on it -- specifically,
to compute $\pi()$ and $\pi^{-1}()$ on the perms  represented
by the $q$-permuters -- in the bit-probe, cell-probe or RAM model.
This is sufficient to compute $\pi()$ and $\pi^{-1}$ in the
$(q,r)$ Benes network at large.

\subsubsection{Representing Small Perms}

In this section we consider the highly space-efficient representation of
``small'' perms to use as a central $q$-permuter in a $(q,r)$-Benes network.
It is straightforward (as noted in Section~\ref{subsec:optimality}) to represent
a perm on $[q]$, $q = O(\lg n/\lg \lg n)$ and operate on it in
the cell-probe model, or by table lookup in the RAM model.
As we will see, the larger
we can make our central $q$-permuters (while keeping optimal
space and reasonable processing times), the lower the redundancy of
our representation.  With this in mind, we now
give a method for asymptotically larger values of $q$.  We
use the following complexity bounds for integer multiplication
and division using the fast Fourier Transform \cite{CLRS}:

\begin{lemma}\label{lem:division}
Given a number $A$ occupying $m$ words and another number $B \le A$,
one can compute the numbers $(A \bmod B)$ and $(A \bdiv B)$ in
$O(m \lg m)$ time.
\end{lemma}

\begin{lemma}\label{lem:small}
If $q \le (\lg n)^2 / (\lg \lg n)^4$, then there is
a representation of an arbitrary perm $\pi$ on $[q]$ using $\calp(q)$
bits that supports $\pi(i)$ and $\pi^{-1}(i)$ in $O(\lg n / \lg\lg n)$ time.
This assumes access to a set of precomputed constants that depend on $q$ and
can be stored in $O(q^2 \lg q)$ bits and also precomputed tables of size
$\sqrt{n} (\lg n)^{O(1)}$ bits.
\end{lemma}

\begin{proof}
We represent a perm $\pi$ over $[q]$ as a sequence $r(0), r(1), \dots, r(q-1)$,
where $r(0) = 0$ and for $1 \le i < q$,
$r(i) = |\{j < i | \pi(j) < \pi(i) \}|$
is the rank of $\pi(i)$ in the set $\{ \pi(0), \pi(1), \dots, \pi(i-1) \}$.
This sequence is viewed as
a $q$-digit number in a ``mixed-radix'' system, where
the $i$-th digit $r(i)$ is from $[i+1]$, representing the
integer $R = \sum_{i=0}^{q-1} i! r(i)$.  The perm
$\pi$ is encoded by storing $R$ in binary:
since $R$ is an integer from $[q!]$, the space used by the encoding
is $\calp(q)$ bits, and $R$ is stored in $m = O(\lg n/(\lg \lg n)^3)$
words. To compute $\pi()$ or $\pi^{-1}()$,
we first decode the sequence $r(0),\ldots,r(q-1)$ from $R$ in
$O(m (\lg m)^2)$ time, and from this seqeunce compute $\pi()$ and
$\pi^{-1}()$ in $O(m \lg m)$ and $O(m)$ time respectively,
for an overall running time of $O(m (\lg m)^2) = O(\lg n/\lg \lg n)$.
We now describe these steps, assuming
for simplicity that $q$ is a power of 2.

To decode $R$, we first obtain representations
$R'$ and $R''$ of the sequences of digits
$r(q-1),r(q-2),\ldots,r(q/2)$, and $r(q/2-1),\ldots,r(0)$ as
$R' = (R \bdiv (q/2)!)$ and $R'' = (R \bmod (q/2)!)$ in $O(m \lg m)$ time,
and recurse. When recursing, note that $\lg R' - (\lg R)/2 = O(q)$ bits, so the
lengths of $R'$ and $R''$ are equal to within $O(m/\lg m)$ words.  Standard
arithmetic, plus table lookup, is used once the integer to be decoded fits
into a single word. Thus, the recurrence is:
\begin{eqnarray*}
 T(m) &=& m \lg m + T(m_1) + T(m_2)\\
 T(1) &=& O(1)
\end{eqnarray*}
where $m_1 + m_2 \le m + 1$ and $|m_j - m/2| = O(m/\lg m)$ (for $j = 1, 2$), which
clearly solves to $O(m (\lg m)^2)$.  (It is assumed that the divisors
at each level of the recursion such as $(q/2)!$ at the top level,
$(q/4)!$ and $(3q/4)(3q/4 - 1)\cdots(q/2)$ at the next level etc.
are pre-computed (but these depend on $q$ only, and are independent
of the perm $\pi$).

We partition
the sequence $r(q-1),\ldots,r(0)$ into \emph{chunks} of
$ c= \lceil \frac{1}{2}(\lg n/\lg q) \rceil$ consecutive numbers each; each
chunk fits into a single word and the number of chunks is $O(m)$.
Define $\Under(x,i)$ as the number of values in
$\pi(q-1),\ldots,\pi(i)$ that are $\le x$.
As $r(q-1) = \pi(q-1)$, $\Under(x,q-1)$ is immediate.
 Further observe that:
\begin{itemize}
\item if $r(i) = x - \Under(x,i+1) - 1$ then $\pi(i) = x$;
\item if $r(i) < x - \Under(x,i+1) - 1$ then $\pi(i) < x$;
\item if $r(i) > x - \Under(x,i+1) - 1$ then $\pi(i) > x$.
\end{itemize}
Thus, $\Under(x,i)$ is easily computed from $\Under(x,i+1)$ and $r(i)$.
Given $\Under(x,i)$ and a chunk $r(i-1),\ldots,r(i-c)$ one can
perform all the following tasks in $O(1)$ time using table lookup:
\begin{itemize}
\item compute $\Under(x,i-c)$;
\item determine if there is a $j \in [i-1,i-c]$ such that
$\pi(j) = x$;
\item given a position $j \in [i-1,i-c]$, determine whether $\pi(j)$
$\le x$ or $> x$.
\end{itemize}
This gives an $O(m)$-time algorithm for computing $\pi^{-1}()$ and
an $O(m \lg m)$-time algorithm for computing $\pi()$ (via binary search).
\end{proof}

\subsubsection{Representing Larger Perms}

We will now use the representation of Lemma~\ref{lem:small},
to represent larger permutations via the Benes network.
We begin by showing:

\begin{proposition}
\label{prop:approx}
 For all integers $p, t \ge 0$, $p \ge t$ there is an integer $p' \ge p$
such that $p' = q 2^\ell$ and $p' < p(1 + 1/t)$, for integers $q$ and $\ell$
where $t < q \le 2t$ and $\ell \ge 0$.
\end{proposition}

\begin{proof}
Take $q$ to be $\ceil{p/2^\ell}$, where $\ell$ is the integer that
satisfies $t < p/2^\ell \le 2t$.  Note that $p' < (p/2^\ell + 1)\cdot 2^r =
p(1 + 2^r/p) < p(1+1/t)$.
\end{proof}

Now we describe the necessary modifications to the Benes network.
Although no new ideas are needed, a little care is needed to minimize
redundancy.

\begin{lemma}
\label{lem:level1}
For any integer $p \le n$, if $p = q 2^r$ for integers $q$ and $r$ such that
$(\lg n)^2/2(\lg\lg n)^4 < q \le (\lg n)^2/(\lg\lg n)^4$
and $r \ge 0$, then there is a representation of an arbitrary
perm $\pi$ on $[p]$ that uses ${\calp}(p) + \Theta((p \lg q)/q)$ bits,
and supports $\pi()$ and $\pi^{-1}()$ in $O(r + \lg n/ \lg\lg n)$ time each.
This assumes access to a pre-computed table of size $O(\sqrt n (\lg n)^c)$
bits that does not depend upon $\pi$, for some constant $c > 0$.
\end{lemma}

\begin{proof}
Consider the $(q,r)$-Benes network that realizes the perm $\pi$, and
represent this network as follows. List all the switch settings of the
outer $2r$ layers of switches as in Proposition~\ref{prop:simple}, and
represent each of the central $q$-permuters using Lemma~\ref{lem:small}.
The representation of Lemma~\ref{lem:small} requires pre-computed
tables of size $O(\sqrt{n} (\lg n)^c)$ bits (for some constant $c > 0$),
which can be shared over all the applications of the lemma.
We now calculate the space used. Note that:
\begin{eqnarray*}
\calp{}(p) & = & p \lg (p/e) + \Theta(\lg p)
           ~ = ~ q 2^r ( r + \lg q - \lg e) + \Theta(\lg p)\\
           & = & q r 2^r + 2^r (q \lg (q/e)) + \Theta(\lg p)
\end{eqnarray*}
By Lemma~\ref{lem:modben} and Lemma~\ref{lem:small} the
space used by the above representation (excluding lookup tables)
is $q r 2^r + 2^r \calp(q) = q r 2^r + 2^r (q \lg (q/e) + \Theta(\lg q)) =
\calp(p) + \Theta((p \lg q)/q)$.

The running time for the queries follows from the fact that we need to
look at $O(r)$ bits among the outer layers of switch settings, and that the
representation of the central $q$-permuter (Lemma~\ref{lem:small})
supports the queries in $O(\lg n/ \lg\lg n)$ time.
\end{proof}

\begin{theorem}
\label{thm:benesmain}
An arbitrary perm $\pi$ on $[n]$ may be represented using
${\calp(n)} + O(n (\lg\lg n)^5 /(\lg n)^2)$ bits, such that $\pi()$ and $\pi^{-1}()$
can both be computed in $O(\lg n/\lg\lg n)$ time.
\end{theorem}
\begin{proof}
Let $t = (\lg n)^3$.  We first consider representing a perm $\psi$ on
$[l]$ for some integer $l$, $t < l \le 2t$.  To do this, we find an
integer $p = l (1+O((\lg\lg n)^4/(\lg n)^2))$ that satisfies the preconditions
of Lemma~\ref{lem:level1}; such a $p$ exists by
Proposition~\ref{prop:approx}.  An elementary calculation shows that
$\calp(p) = \calp(l) (1 + O((\lg\lg n)^4/(\lg n)^2)) =
\calp(l) + O(\lg n (\lg\lg n)^5)$.
We extend $\psi$ to a perm on $[p]$ by setting $\psi(i) =
i$ for all $l \le i < p$ and represent $\psi$.  By
Lemma~\ref{lem:level1}, $\psi$ can be represented using $\calp(p) +
\Theta((p \lg p) (\lg \lg n)^4/(\lg n)^2) = \calp(l) + \Theta(\lg n(\lg \lg n)^5)$
bits such that $\psi()$ and $\psi^{-1}()$ operations are supported in
$O(\lg n/ \lg\lg n)$ time, assuming access to a pre-computed table of size
$O(\sqrt n (\lg n)^c)$ bits, for some constant $c > 0$.

Now we represent $\pi$ as follows.  We choose an $n' \ge n$ such that
$n' = n(1+O(1/(\lg n)^3))$ and $n' = q2^r$ for some integers $q, r$ such
that $t < q \le 2t$.  Again we extend $\pi$ to a perm on $[n']$ by setting
$\pi(i) = i$ for $n \le i < n'$,
and represent this extended perm.  As in
Lemma~\ref{lem:level1}, we start with a $(q,r)$-Benes network that
realises $\pi$ and write down the switch settings of the $2r$ outer
levels in level-order.  The perms realised by the central
$q$-permuters are represented using Lemma~\ref{lem:level1}.
Ignoring any pre-computed
tables, the space requirement is $q r 2^r + 2^r (\calp(q) + \Theta(\lg
n(\lg \lg n)^5))$ bits, which is again easily shown to be $\calp(n') +
\Theta((n' \lg n')/q + 2^r \lg n(\lg \lg n)^5))$ $= \calp(n') +
\Theta(n (\lg \lg n)^5/(\lg n)^2)$ bits.  Finally, as above, $\calp(n')
= (1+O(1/(\lg n)^3)) \calp(n)$, and the space requirement is
$\calp(n) + \Theta(n (\lg \lg n)^5/(\lg n)^2)$ bits.

The running time for $\pi()$ and $\pi^{-1}()$ is clearly $O(\lg n)$.
To improve this to $O(\lg n/\lg \lg n)$, we now explain how to step
through multiple levels of a Benes network in $O(1)$ time, taking care
not to increase the space consumption significantly.  Consider a
$(q,r)$-Benes network and let $t = \floor{\lg\lg n - \lg \lg \lg n} -
1$.  Consider the case when $t \le r$ (the other case is easier), and
consider input number $0$ to the $(q,r)$-Benes network.  Depending
upon the settings of the switches, a packet entering at input $0$ may
reach any of $2^t$ switches in $t$ steps A little thought shows that
the only packets that could appear at the inputs to these $2^t$
switches are the $2^{t+1}$ packets that enter at inputs
$0,1,k,k+1,2k,2k+1,\ldots$, where $k = q2^{r-t}$.  The settings of the
$t2^t$ switches that could be seen by any one of these packets suffice
to determine the next $t$ steps of {\it all\/} of these packets.
Hence, when writing down the settings of the switches of the Benes
network in the representation of $\pi$, we write all the settings of
these switches in $t 2^t \le (\lg n)/2$ consecutive locations.  Using
table lookup, we can then step through $t$ of the outer $2r$ layers of
the $(q,r)$-Benes network in $O(1)$ time.  Since computing the effect
of the central $q$-permuter takes $O(\lg n/\lg \lg n)$ time, we see that the
overall running time is $O(r/t + \lg n/\lg \lg n) = O(\lg n/\lg \lg n)$.
\end{proof}

\subsection{Optimality}
\label{subsec:optimality}

We now consider the optimality of the solutions given in
the previous two sections: specifically, if they achieve the
best possible redundancy for a given query time. 
As noted in Introduction, Golynski \cite[Theorem 17]{golynski-thesis}
has shown that any data structure in the ``black-box'' model
that supports $\pi^{-1}$ in at most $t < n/2$ evaluations of $\pi()$ requires
an index of size $\Omega((n/t) \lg (n/t))$.  This shows the asymptotic
optimality of Theorem~\ref{thm:shortcuts2} for $t = n^{1-\Omega(1)}$.
In the cell probe model, Golynski \cite{golynski-09} shows that:
\begin{lemma}
\label{lem:golynski}
For any data structure
which uses ${\calp}(n) + r$ bits of space to represent a perm over $[n]$ and
supports $\pi()$ and $\pi^{-1}()$ in time $t_f$ and $t_i$ respectively,
such that $\max\{t_f,t_i\} \le (1/16)(\lg n/\lg \lg n)$, it holds that $r = \Omega((n \lg n)/(t_f\cdot t_i))$ bits.
\end{lemma}
This shows that Corollary~\ref{cor1} is optimal for a range of values
of the parameter $t$.  Specficially, there is a constant $c$
(which depends upon the constant within the $O()$ in Corollary~\ref{cor1}
and the value $1/16$ in Lemma~\ref{lem:golynski}) such that the redundancy of
Corollary~\ref{cor1}
is asymptotically optimal for all $t \le c \lg n / \lg \lg n$.
In order to clarify the relationship of
Lemma~\ref{lem:golynski} to the results in Section~\ref{sec:benes} we
have the following proposition:

\begin{proposition}
In the cell probe model with word size $O(\log n)$, a perm $\pi$ non $[n]$ can
be represented as follows:

\begin{itemize}
\item[i.] Both $\pi()$ and $\pi^{-1}()$ can be computed using $2 \lg n / \lg \lg n + O(1)$ probes, and the space used is $\calp(n) + O(n (\lg \lg n)^2 / \lg n)$ bits.
\item[ii.] Both $\pi()$ and $\pi^{-1}()$ can be computed using $(2 + \epsilon) \lg n / \lg \lg n + O(1)$ probes, for any constant $\epsilon > 0$, and the space used is $\calp(n) + O(n (\lg \lg n)^3 / (\lg n)^2)$ bits.
\end{itemize}
\end{proposition}
\begin{proof} In the cell probe model, we note that given
a perm $\pi$ on $[q]$, one can
compute $\pi()$ and $\pi^{-1}$ on a perm $q$ in $O(1 + (q \lg q)/\lg n)$ time,
using ${\calp(q)}$ bits. This is done by
representing $\pi$ implicitly, e.g., as the index of $\pi$ in a
canonical enumeration of all perms on $[q]$, and computing
$\pi()$ and $\pi^{-1}$ by simply reading the entire representation
(which occupies $O(1 + (q \lg q) / \lg n)$ cells).  Two particular values
of $q$
are of interest here: $q_1 = \Theta(\lg n /\lg \lg n)$, when the time
is $O(1)$ probes, and $q_2 = \epsilon (\lg n /\lg \lg n)^2$, for some constant
$\epsilon < 1$,  when the time is at most $\epsilon \lg n /\lg \lg n$ probes.

Using these representations as the central $q$-permuter in
Lemma~\ref{lem:level1}, followed by Theorem~\ref{thm:benesmain},
we note that the number of probes made in the outer layers of the
Benes network is at most $2 \lg n / \lg \lg n$.  By adding the
probes made to the central $q$-permuter (for both $q=q_1$ and $q=q_2$),
we get the numbers of probes claimed. The redundancies are obtained by
straightforward calculation as in Lemma~\ref{lem:level1} and
Theorem~\ref{thm:benesmain}.
\end{proof}

The first of two cases represents the lowest number of probes that we are able
to achieve with our approach.
Although the number of probes is still higher than the
maximum number of probes allowed by Lemma~\ref{lem:golynski}, the redundancy
equals the lowest redundancy provable by Lemma~\ref{lem:golynski}.  However,
with a very small increase in the number of probes, the redundancy drops
considerably (and in fact is lower than that of Theorem~\ref{thm:benesmain}).

\subsection{Supporting Arbitrary Powers}\label{sec:powers1}

We now consider the problem of representing an arbitrary perm $\pi$ to
compute $\pi^k()$ for $k > 1$ (or $k < 1$) more efficiently than by repeated application of $\pi()$
(or $\pi^{-1}()$). Here we develop a succinct structure to support all
powers of $\pi$ (including $\pi()$ and $\pi^{-1}$). The results in this
section assume that we have ${\calp(n)}$ bits (plus some redundancy) to
store the representation, i.e., we do not work in the ``black-box'' model.

\begin{theorem}
\label{mainpower}
Suppose there is a representation $R$ taking $s(n)$ bits to store an
arbitrary perm $\pi$ on $[n]$, that supports $\pi ()$ in time $t_{f}$, and
$\pi^{-1} ()$ in time $t_{i}$. Then there is a representation for an arbitrary
perm on $[n]$ taking $s(n) + O(n \lg n/\lg \lg n)$ bits in which
$\pi^{k} ()$ for any integer $|k| \le n$ can be supported
in $t_{f} + t_{i} + O(1)$ time, and one taking $s(n) + O(\sqrt{n} \lg n)$ bits
in which $\pi^{k} ()$ can be supported in $t_{f} + t_{i} + O(\lg \lg n)$ time.
\end{theorem}

\begin{proof}
Consider the cycle representation of the given perm $\pi$, in which for all
cycles of $\pi$, we write down the elements comprising the cycle, in the
order in which they appear in the cycle, starting with the smallest element in the cycle.  It will be convenient to consider the logical array $\psi$
of length $n$, which comprises the cycles written in nondecreasing order of length, with logical separators marking the boundary of each cycle (see Fig.~\ref{fig:cyclerep} for an example)\footnote{One can dispense with the logical separators by writing the cycles in order of decreasing minimum element, but this is not as convenient for our purposes.}.  Clearly, ignoring
the logical separators between cycles, $\psi$ is itself a permutation.
\begin{figure}
\begin{center}
\epsfxsize 0.8\textwidth \epsfbox{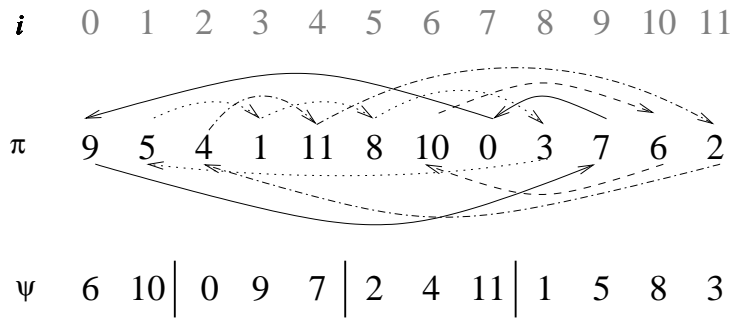}
\end{center}
\caption{A permutation $\pi$ and the logical array $\psi$ representing its cycles.}
\label{fig:cyclerep}
\end{figure}

To compute $\pi^{k}(x)$ for any (positive or negative) $k$ we do the following:
\begin{enumerate}
\item find the position $j$ in $\psi$ that contains $x$,
\item find the left endpoint $l$ of the segment of $\psi$ that represents
the cycle containing $i$, and the length $\lambda$ of this cycle and
\item return the element of $\psi$ in position $s = l + ({{(j-l+k)}\bmod{\lambda}})$.
\end{enumerate}

The data structure for implementing this is as follows.
We represent $\psi$ in the assumed representation $R$.
In Step (1), $j$ is computed as $\psi^{-1}(i)$ in time $t_{i}$, and in
Step (3), the return value is just $\psi(s)$, computed in time $t_{f}$.
We now focus
on Step (2).  Let $\lambda_1 < \lambda_2 < \ldots < \lambda_z$ be
the \emph{distinct} cycle lengths in $\pi$ (the example in
Fig.~\ref{fig:cyclerep} has $z=3$); note that $z = O(\sqrt{n})$.
We store the sequence $\{\lambda_i\}$ in an array, using
$O(\sqrt{n}\lg n)$ bits.  Also consider the set $S = \{s_i\}$,
where $s_1 = 0$ and for $i=2,\ldots,z$, $s_i$ is the
total length of all cycles in $\pi$ whose length is strictly less
than $\lambda_i$ (note that $s_i$ is the starting position of
the sequence of cycles of size $\lambda_i$).  Thus, if $j$ is the
position of $x$ in $\psi$ in Step (1), then the length $\lambda$
of the cycle containing $x$ is $\lambda_t$, where $t = \fullrank(j,S)$.
Also, since all the cycles of length $\lambda$ begin at
$s_t = \Select(S, t)$, it is straightforward to compute the
left endpoint of the cycle containing $x$.  It only remains
to describe how to represent $S$.  We choose two options, giving
the claimed results:

\begin{itemize}
\item to represent $S$ in the FID of Theorem~\ref{litomfid},
taking $\lg {n \choose z} + O(n \lg \lg n/\lg n) = O(n \lg \lg n/\lg n)$ bits,
which supports $\fullrank$ and $\Select$ in $O(1)$ time.
\item to represent $S$ as an array, supporting $\Select$ in $O(1)$
time and also as a predecessor data structure (e.g. the Y-fast trie
\cite{Willard83}) which supports $\fullrank$ in $O(\log \log n)$ time.
The space used by this option is $O(\sqrt{n} \lg n)$ bits. 
\qedhere
\end{itemize}
\end{proof}

As an immediate corollary, we get, from Theorem \ref{thm:benesmain}
\begin{corollary}
There is a representation to store an arbitrary perm $\pi$ on $[n]$
using at most ${\calp(n)} + O(n (\lg \lg n)^5/(\lg n)^2)$
bits that can support $\pi^{k}()$ for any $k$ in $O(\lg n/\lg \lg n)$ time.
\end{corollary}

\section{Succinct trees with level-ancestor queries}\label{sec:level-anc}
In this section we consider the problem of supporting \emph{level-ancestor}
queries on a static rooted ordered tree.
The structure developed here 
will be used in the next section as a substructure in representing a 
function efficiently.
Given a rooted tree $T$ with $n$ nodes, the 
level-ancestor problem is to preprocess $T$ to answer queries of the
following form: Given a vertex $v$ and an integer $i>0$, find the
$i$th vertex on the path from $v$ to the root, if it exists. 
Existing solutions take $\Theta(n \lg n)$ bits to
answer queries in $O(1)$ time \cite{Dietz,BV,AH,BF}, and our
solution stores $T$ using (essentially optimal)
$2n$ bits of space, and uses auxiliary structures of $o(n)$ bits to
support level-ancestor queries in $O(1)$ time.  
Another useful feature of our
solution (which we need in the function representation) is that it
also supports finding the level-successor (or predecessor) 
of a node, i.e., the node
to the right (left) of a given node on the same level, if it exists, in
constant time.

A high-level view of our structure and the query algorithm is as
follows: for any constant $c>0$  
we construct a structure $A$, that given a node $x$ and
any (positive or negative) integer $k$, $|k| \le \lg^c n$, 
supports finding the ancestor (or the first successor in pre-order, if $k \le 0$) 
of $x$ whose depth is $\depth(x) + k$ (this structure is our main contribution).
Applying the above with $c=2$ (say), we also construct another structure, $B$, which supports
level-ancestor queries on nodes whose depths are multiples of $\lg^2
n$, and whose heights are at least $\lg^2 n$. To support a level-ancestor query,
structure $A$ is first used to find the closest ancestor of the given
node, whose depth is a multiple of $\lg^2 n$ and whose height is 
at least $\lg^{2} n$. Then structure $B$ is
used to find the ancestor which is the closest descendant of the
required node and whose depth is a multiple of $\lg^2 n$. Structure $A$
is again used to find the required node from this node. The choice of
different powers of $\lg n$ in the structures given below are somewhat
arbitrary, and could be fine-tuned to slightly improve the lower-order
term.

The structure $A$ consists of
the tree $T$ represented in $2n$ bits as a balanced parenthesis
(BP) sequence as in \cite{MR}, by visiting the nodes of the tree
in depth first order and writing an open parenthesis whenever a node
is first visited, and a closing parenthesis when a node is visited
after all its children have been visited. Thus, each node has exactly
one open and one closing parenthesis corresponding to it. Hereafter,
we also refer a node by the position of either the open or the closing
parenthesis corresponding to it in the BP sequence of the
tree.  We store an existing auxiliary structure of size $o(n)$ bits that
answers the following queries in $O(1)$ time on the BP sequence (see \cite{MR,Geary-BP} for
details):

\begin{itemize}
\item $\close(i)$: find the position of the closing parenthesis that
matches the open parenthesis at position~$i$.
\item $\open(i)$: find the position of the open parenthesis that
matches the closing parenthesis at position~$i$.
\item $\excess(i)$: find the difference between the number of open
parentheses and the number of closing parentheses from the beginning
up to the position~$i$.
\end{itemize}
Note that the $\excess$ of a position $i$ is simply the depth of the
node $i$ in the tree.  Our new contribution is to give a 
$o(n)$-bit structure to support the following operation in $O(1)$ time:

\begin{itemize}
\item $\nextexcess(i,k)$: find the least position $j > i$ such that
$\excess(j) = k$.
\end{itemize}

We only support this query for $\excess(i) - O(\lg^c n) \le k \le 
\excess(i) + O(\lg^c n)$ for some fixed constant $c$. In the following lemma, we 
fix the value of $c$ to be $2$. 
Observe that $\nextexcess(i,k)$ gives:
\begin{itemize}
\item[(a)] the ancestor of $i$ at depth $k$, if $k < \depth(i)$, and 
\item[(b)] the next node after $i$ in the level-order traversal of the tree,
if $k = \depth(i)$, and
\item[(c)] the next node after $i$ in pre-order, if $k > \depth(i)$. 
\end{itemize}

We now describe the auxiliary structure to support the $\nextexcess$
query in constant time using $o(n)$ bits of extra space, showing the following:

\begin{theorem}\label{thm:next-excess} 
Given a balanced parenthesis sequence of length $2n$, one can support 
the operations \open, \close, \excess{} and $\nextexcess(i,k)$ 
where $|k-\excess(i)| \le \lg^{2} n$, all in constant time
using an additional index of size $o(n)$ bits.
\end{theorem}

\begin{proof}
The auxiliary structure to support \open{}, \close{} and \excess{} in constant time 
using $o(n)$ additional bits has been described by Munro and Raman~\cite{MR}
(see also \cite{Geary-BP} for a simpler structure). We now describe the 
auxiliary structures required to support the \nextexcess{} query in constant time.

We split the parenthesis sequence corresponding to the tree into 
superblocks of size $s = \lg^4 n$ and each superblock into blocks 
of size $b = (\lg n)/2$.  
Since the excess values of two consecutive positions differ only by
one, the set containing the excess values of all the positions in a
superblock/block forms a single range of integers, which we denote 
as the {\em excess-range} of the superblock/block. We store this 
excess range information for each superblock, which requires 
$O(n \lg n / \lg^{4} n) = o(n)$ bits for the entire sequence. 
For each block, we also store the excess-range information, where 
$\excess$ is defined with respect to the beginning of the superblock. 
As the excess-range for each block can be stored using $O(\lg\lg n)$
bits, the space used over all the blocks is $O(n \lg\lg n / \lg n) = o(n)$ bits.

For each superblock, we store the following 
structure to support the queries within the superblock (i.e., if the answer 
lies in the same superblock as the query element) in $O(1)$ time:

We build a complete tree with branching factor $\sqrt{\lg n}$ (and
hence constant height) with blocks at the leaves. Each internal node of 
this tree stores the excess ranges of all its children, where the 
excess-range of an internal node is defined as the union of the
excess-ranges of all the leaves in its subtree. 
Thus, the size of this structure for each superblock 
is $O(s \lg\lg n/ b) = o(s)$ bits. Using this structure, given any position
$i$ in the superblock and a number $k$, we can find the position
$\nextexcess(i,k)$ in constant time, if it exists within the superblock.
More specifically, a query is answered by starting at the leaf (block) $v$ 
containing the position $i$, traversing the tree upwards till we find the first 
ancestor node which has a child with preorder number larger than that of 
$v$ whose excess-range contains $k$, and then traversing downwards 
to reach the leaf containing the answer to the query; searches at the 
internal nodes and leaves are performed using precomputed tables, 
as the information stored at these nodes is either $O(\sqrt{\lg n} \lg\lg n)$ 
bits for internal nodes, or $(\lg n)/2$ bits for leaves.

Let $[e_1, e_2]$ be the range of $\excess$ values in a superblock $B$.  
Then for each $i$ such that $e_1 - \lg^2 n \le i < e_1$ or 
$e_2 \le i < e_2 + \lg^2 n$, we store the least position to the right of 
superblock $B$ whose excess is $i$, in an array $A_B$.  

In addition, for each $i$, $e_1 \le i \le e_2$, we store a pointer to
the first superblock $B'$ to the right of superblock $B$ such that $B'$ has a
position with excess $i$. Then we remove all multiple pointers (thus
each pointer corresponds to a range of excesses instead of just one
excess). The graph representing these pointers between superblocks is
planar. [One way to see this is to draw the graph on the Euclidean
plane so that the vertex corresponding to the $j$-th superblock $B$, with
$\excess$ values in the range $[e_1,e_2]$, is represented as a
vertical line with end points $(j,e_1)$ and $(j,e_2)$.  Then, there is
an edge between two superblocks $B$ and $B'$ if and only if the vertices
(vertical lines) corresponding to these are `visible' to each
other (i.e., a horizontal line connecting these two vertical lines at
some height does not intersect any other vertical lines in the
middle).] Since the number of edges in a planar graph on $m$ vertices
is $O(m)$, the number of these inter-superblock pointers (edges) is
$O(n/s)$ as there are $n/s$ superblocks (vertices).  The total space
required to store all the pointers and the array $A_B$ is 
$O(n \lg^3 (n/s)) = o(n)$ bits.

Thus, each superblock has a set of pointers associated with a set of ranges
of excess values. Given an excess value, we need to find the range
containing that value in a given superblock (if the value belongs to the
range of excess values in that superblock), to find the pointer associated
with that range. For this purpose, we store the following auxiliary
structure: If a superblock has more than $\lg n$ ranges associated with it
(i.e., if the degree of the node corresponding to a superblock in the graph
representing the inter-superblock pointers is more than $\lg n$),
then we store a bit vector for that superblock that has a $1$ at the
position where a range starts, and $0$ everywhere else. We also store
an auxiliary structure to support $\rank$ queries on this bit vector
in constant time. Since there are at most $n/(s \lg n)$ superblocks
containing more than $\lg n$ ranges, the total space used for storing
all these bit vectors together with the auxiliary structures is $o(n)$
bits. If a superblock has at most $\lg n$ ranges associated with it, then
we store the lengths of these ranges (from left to right) using the
searchable partial sum structure of \cite{RRR-WADS}, that supports
predecessor queries in constant time. This requires $o(s)$ bits for 
every such superblock, and hence $o(n)$ bits overall.

Given a query $\nextexcess(i,k)$, let $B$ be the superblock to which the
position $i$ belongs. We first check to see if the answer lies within
the superblock $B$ (using the prefix sums tree structure mentioned above),
and if so, we output the position. Otherwise, let $[e_1,e_2]$ be the
range of $\excess$ values in $B$. If $e_1 - \lg^2 n \le k < e_1$ or 
$e_2 \le k < e_2 + \lg^2 n$, then
we can find the answer from the array $A_B$. Otherwise (when $e_1 \le
k \le e_2$), we first find the pointer associated with the range
containing $k$ (using either the bit vector or the partial sum
structure, associated with the superblock) and use this pointer to find the
block containing the answer. Finding the answer, given the superblock in
which it is contained, is done using the prefix sums tree structure
stored for that superblock.

Thus, using these structures, we can support $\nextexcess(i,k)$ for
any $i$ and $|k - \excess(i)| \le \lg^2 n$ in constant time.  
\end{proof}

By using the balanced parenthesis representation 
of the given tree  and by storing the auxiliary structures of Theorem~\ref{thm:next-excess}, 
we can support the following:
given a node in the tree find its $k$-th ancestor, for $k \le \lg^2 n$, and also 
the next node in the level-order traversal of the tree in constant time. 
To support general level ancestor queries, we do as follows.

Firstly, we mark all nodes of the tree that are at a depth which is a multiple of
$\lg^2 n$ and whose height is at least $\lg^2 n$ (similar to~\cite{AH}). 
There are $O(n/\lg^2 n)$ such nodes. 
We store all these marked nodes as a tree (preserving
the ancestor relation among these nodes) and store a linear space
(hence $o(n)$-bit) structure that supports level-ancestor queries in
constant time \cite{BF}. Note that one level in this tree corresponds
to exactly $\lg^2 n$ levels in the original tree. We also store the
correspondence between the nodes in the original tree and those in the
tree containing only the marked nodes.

A query for $\levelancestor(x,k)$, the ancestor of $x$ at height $k$
from $x$ (i.e., at depth $\depth(x) - k$), is answered as follows: If
$k \le \lg^2 n$, we find the answer using a $\nextexcess$ query.
Otherwise, we first find the least ancestor of $x$ which is marked
using at most two $\nextexcess$ queries (the first one to find the
least ancestor whose depth is a multiple of $\lg^2 n$, and the next
one, if necessary, to find the marked ancestor whose height is at least
$\lg^2 n$). From this we find the highest marked ancestor of $x$ which
is a descendant of the answer node, using the level-ancestor structure
for the marked nodes. The required ancestor is found from this node
using another $\nextexcess$ query, if necessary. 

The query $\levelsuccessor(x)$, which returns the successor of node 
$x$ in the level order (i.e., the node to the right of $x$ which is in the 
same level as $x$), can be supported in constant time using a 
$\nextexcess(x,depth(x))$ query.
Since all the nodes in a subtree are together in the
parenthesis representation, checking whether a node $x$ is a
descendant of another node $y$ can be done in constant time by
comparing either the open or closing parenthesis position of $x$ with
the open and closing parenthesis positions of $y$. Hence the representation 
also supports the $\isancestor$ operation in constant time.

Thus we have:

\begin{corollary}\label{cor:level-ancestor}
Given an unlabeled rooted tree with $n$ nodes, there is a structure
that represents the tree using $2n+o(n)$ bits of space and supports
$\parent$, $\firstchild$, $\levelancestor$, $\levelsuccessor$ and
$\isancestor$ queries in O(1) time.
\end{corollary}

\section{Representing functions}\label{sec:functs}
We now consider the representation of functions $f: [n] \rightarrow [n]$.
Given such a function $f$, we equate it to a digraph in which every node is of
outdegree $1$, and represent this graph space-efficiently. We then show how
to compute arbitrary powers of the function by translating them into
the navigational operations on the digraph.

More specifically, given an arbitrary function $f: [n] \rightarrow [n]$, consider the digraph
$G_f = (V,E)$ obtained from it, where $V = [n]$ and $E = \{
\angle{i,j} : f(i) = j \}$. In general this digraph consists of a set
of connected components where each component has a directed cycle with
each vertex being the root of a (possibly single node) directed tree,
with edges directed towards the root. See Figure
\ref{fig:function-graph}(a) for an example. We refer to each connected component as
a \emph{gadget}.

\begin{figure}[h]
\subfigure[Graph representation of the function $f(x) = (x^2 + 2x
-1) \bmod 19$, for $0 \le x \le 18$. The vertex labels in the
brackets correspond to the function $g$ obtained by renaming the
vertices] { \epsfxsize \textwidth
\epsfbox{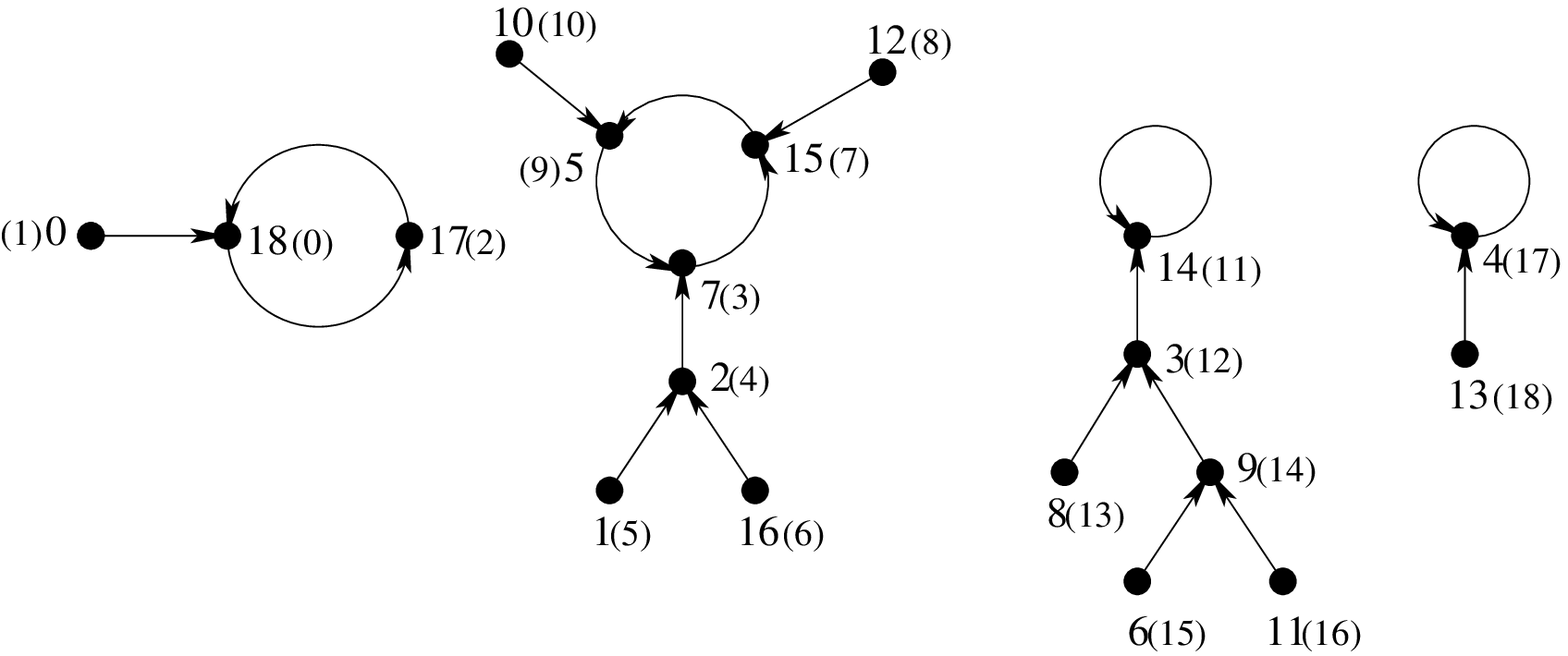}}
%
\subfigure[Perm defining the isomorphism between $G_f$ and $G_g$] {
\vbox{\hsize=\textwidth
\renewcommand{\tabcolsep}{5pt}
  \def\sp{\hspace{0cm}}
  \begin{center}
\begin{tabular}{ccccccccccccccccccc}
0\sp &1\sp &2\sp &3\sp &4\sp &5\sp &6\sp &7\sp &8\sp &9\sp &10\sp %
     &11\sp &12\sp &13\sp &14\sp &15\sp &16\sp &17\sp &18\\
1\sp &5\sp &4\sp &12\sp &17\sp &9\sp &15\sp &3\sp &13\sp &14\sp %
     &10\sp &16\sp &8\sp &18\sp &11\sp &7\sp &6\sp &2\sp &0\\
\end{tabular}
\end{center}}}
\subfigure[Parenthesis representation and the bit vectors indicating the
starting positions of the gadgets and the trees (auxiliary structures are
not shown)]
{ \vbox{\hsize=\textwidth

\renewcommand{\tabcolsep}{1.8pt}
\begin{tabular}{cccclcclcccccccclcccclcccclcccccccccccclcccc}
( &( &) &) &  &( &) & & ( &( &( &) &( &) &) &) & &( &( &) &) & &( &( &) &) & &( &( &( &) &( &( &) &( &) &) &) &) & &( &( &) &) \\
1 &0 &0 &0 &  &0 &0 & & 1 &0 &0 &0 &0 &0 &0 &0 & &0 &0 &0 &0 & &0 &0 &0 &0 & &1 &0 &0 &0 &0 &0 &0 &0 &0 &0 &0 &0 & &1 &0 &0 &0 \\
1 &0 &0 &0 &  &1 &0 & &1  &0 &0 &0 &0 &0 &0 &0 & &1 &0 &0 &0 & &1 &0 &0 &0 & &1 &0 &0 &0 &0 &0 &0 &0 &0 &0 &0 &0 & &1 &0 &0 &0 \end{tabular}}}
\caption{Representing a function}
\label{fig:function-graph}
\end{figure}

The main idea of our representation is to store the structure of the
graph $G_f$ as a tree $T_{f}$ such that the forward and inverse queries
can be translated into appropriate navigational operations on the tree.
We store the bijection between the nodes labels in $G_{f}$ and the preorder
numbers of the `corresponding' nodes in $T_{f}$ as a perm $\pi$.
To support the queries for powers of $f$, we need to
find the node in $T_{f}$ corresponding to a given label, perform the
required navigational operations on the tree to find the answer node(s), and
finally return the label(s) corresponding to the answer node(s).
Hence we store the perm $\pi$ using one of the perm representations
from Section~\ref{sec:perms} so that $\pi()$ and $\pi^{-1}()$ can be supported
efficiently.

We define a gadget to be {\em wide} if its cycle length is larger than
$\lg^{1/3} n$, and {\em narrow} otherwise. The {\em size} of a gadget
or a tree is defined as the number of nodes in it.
Before constructing the tree $T_{f}$, we first re-order the gadgets and the tree
nodes within each gadget as follows:
(i) We first order the gadgets so that all the narrow gadgets are before
any of the wide gadgets.
(ii) Wide gadgets are ordered arbitrarily among themselves, while
narrow gadgets are ordered in the non-decresing order of their sizes.
(iii) Within each group of narrow gadgets with the same size, we arrange
them in the non-decreasing order of their cycle lengths (the cycle
length of a gadget is the number of trees in the gadget).
(iv) For each gadget whose cycle length is greater than $1$,
we break the cycle by selecting a tree with maximal height among
all the tree that belong to the gadget and deleting the outgoing edge
from the root of this tree. We then order the trees such that the trees
are in the reverse order as we move along the cycle edges in the forward
direction (thus the tree with the maximal height that was selected,
is the last tree in this order).
(v) We also arrange the nodes within each tree such that the
leftmost path of any subtree is the longest path in that subtree,
breaking the ties arbitrarily.

We now construct a tree that encodes the structure of the function $f$.
Let $C_1, C_2, \dots, C_p$ be the gadgets in $G_f$ and let
$T_i^1, T_i^2, \dots, T_i^{q_{i}}$ be the trees in the $i$-th gadget, for
$1 \le i \le p$, after the re-ordering of the gadgets and the nodes the
within the trees. Let $r_{i}^{j}$ be the root of the tree
$T_{i}^{j}$, for $1 \le i \le p$ and $1 \le j \le q_{i}$. We refer the node
$r_{i}^{1}$ as the {\em root of the gadget} $C_{i}$.

Construct a tree $T_{f}$ with root $r$ whose children are the $p$ nodes:
$r_{1}^{1}, r_{2}^{1}, \dots r_{p}^{1}$. For $1 \le i \le p$, under the node
$r_{i}^{1}$ add the path $r_{i}^{2} - r_{i}^{3} - \dots - r_{i}^{q_{i}}$. Also attach
the subtree under the root $r_{i}^{j}$ in $T_{i}^{j}$ to the node $r_{i}^{j}$ in $T_{f}$.
The size of $T_{f}$ is $n+1$ (the $n$ nodes in $G_{f}$ plus the new root $r$).
We represent the tree $T_{f}$ using the structure of Corollary~\ref{cor:level-ancestor}
using $2n+o(n)$ bits. Items (iv) and (v) above ensure that the leftmost path in any
subtree of $T_{f}$ is a longest path in that subtree, and hence is represented by a
sequence of open parentheses in the BP sequence. This enables us to
find the descendent of any node in the subtree at a given level, if it exists,
in constant time.

We number of the nodes of $T_{f}$ with their pre-order numbers,
starting from $0$ for the root $r$.
Every node in the tree $T_{f}$, except for the root $r$,
corresponds to a unique node in the graph $G_{f}$, and
this correspondence can be easily determined from the
construction of the tree.
As mentioned earlier, we store this bijection $\pi$ between the labels in
$G_{f}$ and the preorder numbers in $T_{f}$ by representing the perm $\pi$
that supports $\pi()$ and $\pi^{-1}()$ efficiently.

In addition to the perm $\pi$ and the tree $T_{f}$, we store the following
data structures using $o(n)$ bits:
\begin{enumerate}
\item An array $A$ storing the distinct sizes of the narrow gadgets in the
increasing order (i.e., the sequence $s_{1}, s_{2}, \dots, s_{d}$,
where $1 \le s_{1} < s_{2} < \dots < s_{d} \le n$, and for $1 \le i \le d$
there exists a narrow gadget  of size $s_{i}$ in $G_{f}$). Note than $d = O(\sqrt{n})$.
\item An FID for the set $B = \{ p_{1}, p_{2}, \dots p_{d} \}$, where $p_{i}$
is the preorder number of the first narrow gadget (in the above ordering)
whose size is $s_{i}$ (or equivalently, the sum of the sizes of all the narrow
gadgets in $G_{f}$ whose sizes are less than $s_{i}$), for $1 \le i \le d$.
\item An FID for the multiset $C = \{   s_{i,j} \}$, for $1 \le i \le d$ and $1 \le j \le n^{1/3}$,
where $s_{i,j}$ is the sum of the sizes of all the gadgets whose sizes are:
(i) less than $s_{i}$, and (ii) equal to $s_{i}$ whose cycle lengths are at most $j$.
(A $\fullrank$ operation in this FID enables us to find the cycle length of the gadget
containing the node with a given preorder number, if it is in a narrow gadget).
\item An array $A'$ that stores the size and cycle length of each wide gadget,
in the above ordering of the wide gadgets.
\item An FID for the set $B' = \{ p'_{1}, p'_{2}, \dots p'_{d'} \}$, where $d'$ is the
number of wide gadgets in $G_{f}$, and $p'_{i}$ is the preorder number of the
root of the $i$-th wide gadget (in the above ordering).
\end{enumerate}

Given a node in a tree, we can find its $k$-th successor (i.e., the
node reached by traversing $k$ edges in the forward direction), if it
exists within the same tree, in constant time using a $\levelancestor$
query.  The $k$-th successor of node $r_{i}^{j}$ (the root of the $j$th
tree in the $i$th gadget) can be found in $O(1)$ time by computing the
length of the cycle in the $i$th gadget, using $\fullrank$ and
$\select$ operations on the the above FIDs.  By combining these two,
we can find the $k$-th successor of an arbitrary node in a gadget in
constant time.

Given a node $x$ in a gadget, if it is not the root of any tree, then
we can find all its $k$-th predecessors (i.e., all the nodes reachable
by traversing $k$ edges in the reverse direction) in optimal time
using the tree structure by finding all the descendant nodes of $x$
that are $k$ levels below, as follows:
we first find the leftmost descendant in the subtree rooted at $x$ at
the given level, if it exists, in constant time, as the leftmost path is
represented by a sequence of open parentheses in the parenthesis
representation of the tree. From this node, we can find all the nodes
at this level by using the $\levelsuccessor$ operation to find the next
node at this level, checking whether the node is a descendant of $x$
using the $\isancestor$ operation, and stopping when this test fails.

To report the set of all $k$-th predecessors of a node $r_{i}^{j}$ (which is the root of the
$j$th tree in the $i$th gadget), if $j+k \le q_{i}$, then we report all the nodes in the subtree
(of $T_{f}$) rooted at $r_{i}^{j}$ that are at the same level as $r_{i}^{j+k}$. Otherwise,
we first find all trees $T_{x}^{y}$ which contain at least one answer, and then report all
the answers in each of those trees.

Now to find all the trees $T_{i}^{j}$ that contain at least one answer, we observe
that if $T_{i}^{j'}$ contains at least one node that is a $k$-th predecessor of $r_{i}^{j}$,
then it also contains at least one node that is a $(q_i + (k \bmod q_i))$-th predecessor
of $r_{i}^{j}$ (here $q_{i}$ is the number of trees in the $i$th gadget).
Also, the set of all $(q_i + (k \bmod q_i))$-th predecessors of $r_{i}^{j}$ is a subset of
the set of $k$-th predecessors of $r_{i}^{j}$, when $k \ge q_{i}$. In other words, the
set of all trees that contain at least one $k$-th predecessor of $r_{i}^{j}$
is the same as the set of all trees that contain at least one $(q_i + (k \bmod q_i))$-th
predecessor of $r_{i}^{j}$.

Thus to find the $k$-th predecessors of $r_{i}^{j}$, we identify two subsets of trees
whose union is the set of all trees in the gadget $C_i$ that contain at least one
answer. These two subsets are the set of all trees that contain at least one node
\begin{itemize}
\item at a depth of $k$ in the subtree rooted at node $r_{i}^{j}$ in $T_{f}$, and
\item at a depth of $k - (q_{i} - j)$ in the subtree rooted at $r_{i}^{1}$ in $T_{f}$.
\end{itemize}
Once we identify all the trees containing at least one answer, we can report all
the answer nodes in the tree $T_{f}$ in time linear in the number of such
nodes, as explained earlier. Each of these node numbers are then transformed
into their corresponding node numbers in $G_{f}$ using the representation of $\pi$.

Combining all these, we have:

\begin{theorem}
If there is a representation of a perm on $[n]$ that takes $P(n)$
space and supports forward in $t_f$
time and inverse in $t_i$ time, then there is a representation of a
function $f: [n] \rightarrow [n]$ that takes $P(n) + 2n + o(n)$ bits of space
and supports $f^k(i)$ in $O(t_f + t_i * |f^k(i)|)$ time (or in $O(t_i + t_f * |f^k(i)|)$ time),
for any integer $k$ (which can be stored in $O(1)$ words) and for any $i \in [n]$.
\end{theorem}

Using the succinct perm representation of
Corollary \ref{cor1}, we get:

\begin{corollary}
There is a representation of a function $f: [n] \rightarrow [n]$ that
takes $(1+\epsilon) n \lg n + O(1)$ bits of space for any fixed
positive constant $\epsilon$, and supports $f^k(i)$ in $O(1+|f^k(i)|)$
time, for any integer $k$ (which can be stored in $O(1)$ words and
for any $i \in [n]$.
\end{corollary}

\subsection{Functions with arbitrary ranges}

So far we considered functions whose domain and range are the same set
$[n]$. We now consider functions $f: [n] \rightarrow [m]$ whose domain and
range are of different sizes, and deal with the two cases:
(i) $n > m$ and (ii) $n < m$ separately. These results can
be easily extended to the case when neither the domain nor the range
is a subset of the other. We only consider the queries for positive
powers.

\noindent
{\bf Case (i) $n > m$: } A function $f: [n] \rightarrow [m]$, where $n > m$
can be represented by storing the restriction of $f$ on $[m]$
using the representation mentioned in the previous section,
together with the sequence $S = f(m+1), f(m+2), \ldots, f(n)$ stored in an
array. This gives a representation that supports forward queries efficiently.

To support the inverse queries, we store the sequence
$S$ using a representation that supports
$\access$ and $\select$ queries efficiently, where $\access(i)$
returns the value $f(m+i)$, and $\select(j,k)$ returns the $k$-th
occurrence of the value $j$ in the sequence.
We use the following representation which is implicit
in Golynski~et~al.~\cite{Golynskietal}:
A sequence $S$ of length $n$ from an alphabet of size $k$
(where $n \geq k$) can be represented as a collection of
$\lceil{n/k}\rceil$ perms over $[k]$ together with $O(n)$
bits such that a $\select$ or an $\access$ query on $S$ can be
answered by performing a single $\pi()$ or $\pi^{-1}$ query on
one of the perms, together with a constant amount of
computation.

In addition, we augment the directed graph $G_{f}$, representing the
function $f$ restricted to $[m]$, with {\em dummy} nodes as follows:
if $f(m+i) = j$, then we add a dummy node $v$ as a `child' of the node
corresponding to $j$ in $G_{f}$. The node $v$ is a {\em representative}
of the set $\{ i | f(i) = j, i > m \}$. We represent this augmented directed
graph to support the forward and inverse queries, using $O(m)$ bits.
We also represent the perm that maps the `real' (non-dummy)
nodes to their original values in the function $f$. Finally, we store an
FID that indicates the positions of the dummy nodes in the order
determined by the representation of $G_{f}$, using $O(m)$ bits (note
that the size of the graph $G_{f}$ is $O(m)$).

To answer a query $f^{k}(i)$ for $i \in [n]$ and $k \ge 1$, we first find
the node $v$ corresponding to $i$ in the augmented graph $G_{f}$.
The node $v$ is a `real' node if $i \le m$, and can be found using the
perm $\pi$ that maps the nodes of $G_{f}$ to their values in $f$ and
the FID indicating the positions of dummy nodes.
We then find the node $u$ that is reached by traversing $k$
edges in the forward direction, using the structure of $G_{f}$. Finally,
the value corresponding to the node $u$ is obtained using the
perm $\pi$. If $i > m$, then the node $v$
is a dummy node, and we can find $j = f(i)$ using an $\access$ query
on the string $S$, and use the fact that $f^{k}(i) = f^{k-1}(j)$ to compute
the answer.

To answer a query $f^{-k}(i)$ for $i \in [m]$ and $k \ge 1$, we first find
the node corresponding to the value $i$ in $G_{f}$, find all the nodes
that can be reached by traversing $k$ edges in the backward direction,
and return the values corresponding to all such nodes.
Thus we have:

\begin{theorem}
If there is a representation of a perm on $[n]$ that takes $P(n)$
space and supports forward in $t_f$
time and inverse in $t_i$ time, then there is a representation of a
function $f: [n] \rightarrow [m]$, $n \ge m$ that takes
$(n-m) \ceil{\lg m} + P(m) + O(m)$ bits of space and supports $f^k(i)$ in
$O(t_f + t_i)$ time, for any positive integer $k$ and for any $i \in [n]$.
There is another representation of $f$ that takes $\ceil{n/m} P(m) + O(m)$
bits that supports, for any $k \ge 1$, $f^{k}(i)$ in $O(t_{f}+t_{i})$ time, and
$f^{-k}(i)$ in $O(t_{f} + t_i * |f^{-k}(i)|)$ time (or in $O(t_i + t_f * |f^{-k}(i)|)$ time).
\end{theorem}

\noindent
{\bf Case(ii) $n < m$: } For a function $f: [n] \rightarrow [m]$, where $n <
m$, larger powers (i.e., $f^k(i)$ for $k \ge 2$) are not defined in
general (as we might go out of the domain after one or more applications
of the function).

Let $R$ be the set of all elements in the range $[m]$ that have pre-images in
the domain $[n]$ whose values are greater than $n$. In the graph $G_{f}$
representing the function $f$, each element in $R$ corresponds to the root
of a tree with no outgoing edges. We order these trees such that elements
corresponding to these roots are in the increasing order. We then store an
indexable dictionary for the set $R \subseteq [m]$ using
$\lg {m \choose |R|} + o(|R|) + O(\lg\lg m)$ bits . Since $|R| \le n$,
this space is at most $n \lg (m/n) + O(n + \lg\lg m)$ bits. The size of the graph $G_{f}$
is $O(n)$ and hence is stored in $O(n)$ bits using the representation described
in the previous section. Finally, we store the correspondence between the node
numbering given by the $O(n)$-bit representation and the actual node labels in
$G_{f}$, except for the nodes corresponding to $R$. As all these nodes are in
the set $[n]$, we need to store a perm $\pi$ over $[n]$.

A query for $f^k(i)$, for $i \in [n]$ and $k \ge 1$ is answered by first
finding the node corresponding to $i$ in $G_{f}$ using $\pi$, then finding
the $k$-th node in the forward direction, if it exists, using the structure
of $G_{f}$, and finally finding the element corresponding to this node,
using the representation of $\pi$ again.
To find the set $f^{-k}(i)$, for $i \in [m]$ and $k \ge 1$, we first find the
node $x$ corresponding to $i$ in $G_{f}$ using either the representation
of $\pi$ if $i \le n$, or using the indexable dictionary stored for the set
$R$ if $n < i \le m$. We then find all the nodes reachable from $x$ by
taking $k$ edges in the backward direction. We finally report the elements
corresponding to each of these nodes, using the representation of $\pi$.
Thus we have:

\begin{theorem}
If there is a representation of a perm on $[n]$ that takes $P(n)$
space and supports forward in $t_f$
time and inverse in $t_i$ time, then there is a representation of a
function $f: [n] \rightarrow [m]$, $n < m$ that takes $n \lg (m/n) + P(n)
+ O(n)$ bits. For any positive integer $k$, this representation supports
the queries for $f^k(i)$, for any $i \in [n]$ (returns the power if
defined and $-1$ otherwise) in $O(t_f + t_i)$ time, and supports
$f^{-k}(i)$, for any $i \in [m]$ in $O(t_{f} + t_i * |f^{-k}(i)|)$ time
(or in $O(t_i + t_f * |f^{-k}(i)|)$ time).
\end{theorem}

\end{document}